\documentclass[10pt,conference]{IEEEtran}
\usepackage{epsf,psfrag,amssymb,amsfonts,cite,subfigure}
\newtheorem{theorem}{Theorem}
\newtheorem{example}{Example}

\newtheorem{definition}{Definition}

\def\psfancypar#1#2{\begingroup\def\par{\endgraf\endgroup\lineskiplimit=0pt}
               \setbox2=\hbox{\large\sc #2}
               \newdimen\tmpht \tmpht \ht2 \advance\tmpht by \baselineskip
               \font\hhuge=Times-Bold at \tmpht
               \setbox1=\hbox{{\hhuge #1}}
               \count7=\tmpht \count8=\ht1
               \divide\count8 by 1000 \divide\count7 by \count8 
               \tmpht=.001\tmpht\multiply\tmpht by \count7 
               \font\hhuge=Times-Bold at \tmpht
               \setbox1=\hbox{{\hhuge #1}}
               \noindent
                \hangindent1.05\wd1
               \hangafter=-2 {\hskip-\hangindent
               \lower1\ht1\hbox{\raise1.0\ht2\copy1}%
                \kern-0\wd1}\copy2\lineskiplimit=-1000pt}

\newcommand{\beq}{\begin{equation}}
\newcommand{\eeq}{\end{equation}}
\newcommand{\bqa}{\begin{eqnarray}}
\newcommand{\eqa}{\end{eqnarray}}
\newcommand{\bqn}{\begin{eqnarray*}}
\newcommand{\eqn}{\end{eqnarray*}}

\newcommand{\be}{\begin{enumerate}}
\newcommand{\ee}{\end{enumerate}}
\newcommand{\bi}{\begin{itemize}}
\newcommand{\ei}{\end{itemize}}
\newcommand{\bd}{\begin{description}}
\newcommand{\ed}{\end{description}}
\newcommand{\ba}{\begin{array}}
\newcommand{\ea}{\end{array}}
\newcommand{\bde}{\begin{definition}}
\newcommand{\ede}{\end{definition}}
\newcommand{\bex}{\begin{example}}
\newcommand{\eex}{\end{example}}

\def\boxit#1{\vbox{\hrule\hbox{\vrule\kern3pt
        \vbox{\kern3pt#1\kern3pt}\kern3pt\vrule}\hrule}}

\def\reals{ { {\rm  I \kern-0.15em R }  } }
\def\complex{ {\,{{\rm C} \kern-0.50em \raise0.20ex {  |}}\, }}

\def\0bf{{\bf 0}}
\def\1bf{{\bf 1}}
\def\2bf{{\bf 2}}
\def\3bf{{\bf 3}}
\def\4bf{{\bf 4}}
\def\5bf{{\bf 5}}
\def\6bf{{\bf 6}}
\def\7bf{{\bf 7}}
\def\8bf{{\bf 8}}
\def\9bf{{\bf 9}}

\def\ybf{{\bf y}}

\def\ybf{{\bf y}}

\def\Rbf{{\bf R}}

\def\Cmat{\mathcal{C}}

\def\Mmat{\mathcal{M}}

\def\Rmat{\mathcal{R}}

\def\Umat{\mathcal{U}}

\def\Xmat{\mathcal{X}}
\def\Ymat{\mathcal{Y}}

\def\Rxx{\Rbf_{\ssstyle X\kern-.1em X}}

\let\ssstyle=\scriptscriptstyle

\def\Kout{\setbox1=\hbox{\Huge\bf K}\hbox to
1.05\wd1{\hspace{.05\wd1}
}}
\def\Sout{\setbox1=\hbox{\Huge\bf S}\hbox to 1.05\wd1{\hspace{.05\wd1}
}}

\title{\LARGE On the Sum Capacity of the Discrete Memoryless Interference Channel with One-Sided Weak Interference and Mixed Interference}\author{\authorblockN{Fangfang Zhu and Biao Chen}
\authorblockA{Syracuse University\\
Department of EECS\\
Syracuse, NY 13244\\
Email: fazhu\{bichen\}@syr.edu}}

\begin{document}
\maketitle
\begin{abstract}
The sum capacity of a class of discrete memoryless interference channels is determined. This class of channels is defined analogous to the Gaussian Z-interference channel with weak interference; as a result, the sum capacity is achieved by letting the transceiver pair subject to the interference communicates at a rate such that its message can be decoded at the unintended receiver using single user detection. Moreover, this class of discrete memoryless interference channels is equivalent in capacity region to certain discrete degraded interference channels. This allows the construction of a capacity outer-bound using the capacity region of associated degraded broadcast channels. The same technique is then used to determine the sum capacity of the discrete memoryless interference channel with mixed interference. The above results allow one to determine sum capacities or capacity regions of several new discrete memoryless interference channels.
\end{abstract}

\section{Introduction}\label{sec:intro}
The interference channel (IC) models the situation where the transmitters communicate
with their intended receivers while generating interference to unintended
receivers. Despite decades of intense research, the capacity region of IC remains unknown except for a few
special cases. These include interference channels with strong and very strong interference \cite{Carleial:75IT, Sato:78IT,Sato:81IT, Han&Kobayashi:81IT, Costa&ElGamal:87IT}; classes of deterministic and semi-deterministic ICs \cite{ElGamal&Costa:82IT, Chong&Motani:09IT}; and classes of discrete degraded ICs \cite{Benzel:79IT, Liu&Ulukus:08IT}.

Parallel capacity results exist for the discrete memoryless IC (DMIC)
and the Gaussian IC (GIC). Carleial first obtained capacity region for GIC with very strong interference \cite{Carleial:75IT}. This result was subsequently extended by Sato \cite{Sato:78IT} to that of DMICs with very strong interference. Note that the definition of DMIC with very strong interference can actually be broadened to be more consistent with its Gaussian counterpart \cite{Xu-etal:Globecom2010}. Sato \cite{Sato:81IT} and Han and Kobayashi \cite{Han&Kobayashi:81IT} independently established in 1981 the capacity region of GIC with strong interference, where the capacity is the same as that of a compound multiple access channel. In \cite{Sato:81IT} Sato also conjectured the conditions of DMICs under strong interfernce, which was eventually proved by Costa and El Gamal \cite{Costa&ElGamal:87IT} in 1987.

While the capacity region for the general GIC remains unknown, there have been recent progress in characterizing the sum capacity of certain GICs, including: GICs with one-sided weak interference \cite{Sason:04IT}, noisy interference  \cite{Shang-etal:09IT, Motahari&Khandani:09IT, Annapureddy&Veeravalli:09IT}, and mixed interference \cite{Motahari&Khandani:09IT}. This paper attempts to derive parallel sum capacity results for DMICs with weak one-sided and mixed inference. Our definitions of one-sided, weak, or mixed interference are motivated by properties associated with the corresponding Gaussian channels. Some of those definitions are intimately related to those introduced in \cite{Chong-etal:07IT} which studies the capacity region of the discrete memoryless Z-channel.

The rest of the paper is organized as follows. Section~\ref{sec: preliminaries} presents the channel model and relevant previous results. Section~\ref{sec:main} defines DMICs with one-sided weak interference and derives their sum capacities. We refer to those DMICs with one-sided interference as DMZIC (i.e., discrete memoryless Z interference channel) for ease of presentation. The equivalence between the DMZIC with weak interference and the discrete degraded interference channel (DMDIC) is established which allows one to construct a capacity outer-bound for the DMZIC using the capacity region of the associated degraded broadcast channel. Section~\ref{sec:mixed} defines DMICs with mixed interference  and derives the sum capapcity for this class of channels. Section~\ref{sec:conclusion} concludes this paper.

\section{Preliminaries}\label{sec: preliminaries}
\subsection{Discrete Memoryless Interference Channels\label{sec:DMIC}}
A DMIC is specified by its input alphabets $\Xmat_1$ and $\Xmat_2$,  output alphabets $\Ymat_1$ and $\Ymat_2$, and the channel transition matrices:
\begin{eqnarray}
p(y_1|x_1x_2)&=&\sum_{y_2\in \Ymat_2}p(y_1y_2|x_1x_2),\label{eq:DMIC_condition1}\\
p(y_2|x_1x_2)&=&\sum_{y_1\in \Ymat_1}p(y_1y_2|x_1x_2).\label{eq:DMIC_condition2}
\end{eqnarray}

The DMIC is said to be \textit{memoryless} if
\begin{equation}
p(y_1^ny_2^n|x_1^nx_2^n)=\prod_{i=1}^n p(y_{1i}y_{2i}|x_{1i}x_{2i}).
\end{equation}

A $(n,2^{nR_1},2^{nR_2},\lambda_1,\lambda_2)$ \textit{code} for a DMIC with independent information consists of two message sets $\Mmat_1 =\{1,2,\cdots,2^{nR_1}\}$ and $\Mmat_2=\{1,2,\cdots,2^{nR_2}\}$ for senders $1$ and $2$ respectively, two encoding functions:
\bqn
f_1: \Mmat_1\rightarrow \Xmat_1^n,\hspace{.1in} f_2: \Mmat_2\rightarrow\Xmat_2^n,
\eqn
two decoding functions:
\bqn
\varphi_1: \Ymat_1^n\rightarrow \Mmat_1,\hspace{.1in} \varphi_2: \Ymat_2^n\rightarrow \Mmat_2,
\eqn
and the average probabilities of error:
\small
  \bqn \lambda_1\!\!\!\!\!\!&=&\!\!\!\!\!\!\frac{1}{|\Mmat_1||\Mmat_2|}\sum_{w_1=1}^{2^{nR_1}}\sum_{w_2=1}^{2^{nR_2}}Pr\{\varphi_1(\ybf_1)\neq w_1|W_1=w_1,W_2=w_2\},\\
  \lambda_2\!\!\!\!\!\!&=&\!\!\!\!\!\!\frac{1}{|\Mmat_1||\Mmat_2|}\sum_{w_1=1}^{2^{nR_1}}\sum_{w_2=1}^{2^{nR_2}}Pr\{\varphi_2(\ybf_2)\neq w_2|W_1=w_1,W_2=w_2\}.
  \eqn
  \normalsize

A rate pair $(R_1,R_2)$ is said to be \textit{achievable} for the DMZIC if and only if there exist a sequence of
$(2^{nR_1}, 2^{nR_2}, n, \lambda_1, \lambda_2)$ codes such that $\lambda_1,\lambda_2\rightarrow 0$ as $n\rightarrow \infty$.
The \textit{capacity region} of a DMZIC is defined as the closure of the set of all
achievable rate pairs.

\subsection{Existing Results for GICs}
The received signals of a GIC in its standard form are
\bqa
Y_1&=& X_1+\sqrt{a}X_2+Z_1, \label{eq:y1}\\
Y_2&=& \sqrt{b}X_1+X_2+Z_2, \label{eq:y2} \eqa
where $a$ and $b$ are
the channel coefficients corresponding to the interference links,
$X_i$ and $Y_i$ are the transmitted and received signals, and the
channel input sequence $X_{i1}, X_{i2}, \cdots, X_{in}$ is subject
to the power constraint $\displaystyle\sum_{j=1}^{n}X_{ij}^2\leq
nP_i$, $i=1,2$, $Z_1$ and $Z_2$ are Gaussian noises with zero mean
and unit variance, independent of $X_1, X_2$.

Sason in \cite{Sason:04IT} proved that the sum capacity for GICs with one-sided weak interference ($a<1$ and $b=0$) is
\bqn
\frac{1}{2}\log(1+P_2)+\frac{1}{2}\log\left(1+\frac{P_1}{1+aP_2}\right).
\eqn

Motahari and Khandani in \cite{Motahari&Khandani:09IT} established that the sum capacity for GICs with mixed interference ($a\leq 1$ and $b\geq 1$) is
\small
\bqn
\min\left\{\frac{1}{2}\log\left(1+\frac{P_1}{1+aP_2}\right),\frac{1}{2}\log\left(1+\frac{bP_1}{1+P_2}\right)\right\} +\frac{1}{2}\log(1+P_2).
\eqn
\normalsize
We attempt to extend these results to DMICs with appropriately defined one-sided weak interference and mixed interference.
\subsection{Useful Properties of Markov Chains}
The following properties of Markov chains are useful throughout the paper:
\begin{itemize}
  \item Decomposition: $X-Y-ZW\Longrightarrow X-Y-Z$;
  \item Weak Union: $X-Y-ZW\Longrightarrow X-YW-Z$;
  \item Contraction: $(X-Y-Z)$ and $(X-YZ-W)\Longrightarrow X-Y-ZW$.
\end{itemize}

\section{The DMZIC with Weak Interference}\label{sec:main}
\subsection{Discrete Memoryless Z-Interference Channel}
\begin{definition}
For the DMIC defined in Section~\ref{sec:DMIC}, if
\bqa
p(y_2|x_2)=p(y_2|x_1x_2),\label{DMZIC_condition2}
\eqa
for all $x_1$, $x_2$, $y_2$, or equivalently,
\bqa
X_1- X_2 - Y_2
\eqa
forms a Markov chain, this DMIC is said to have one-sided interference.
\end{definition}

We refer to such DMIC as simply DMZIC. The definition is a natural extension of that for Gaussian ZIC where $X_2$ causes interference on $Y_1$.
From the definition, it follows that $X_1$ and $Y_2$ are independent for all input distribution $p(x_1)p(x_2)$.

To define DMZIC with weak interference, we first revisit some properties of Gaussian ZIC with weak interference.
Similar to that established in \cite{Chong-etal:07IT}, it is straightforward to show that a Gaussian ZIC with weak interference is equivalent in its capacity region to a degraded Gaussian ZIC
satisfying the Markov chain
\bqa\label{eq:weakcond2}
X_2- (X_1,Y_2)- Y_1.
\eqa
This is referred in \cite{Chong-etal:07IT} as degraded Gaussian Z channel of type-I.
This motivates us to define DMZIC with weak interference as follows.
\begin{definition}
A DMZIC is said to have \textit{weak interference} if the channel transition probability factorizes as
\bqa
p(y_1y_2|x_1x_2)=p(y_2|x_2)p'(y_1|x_1y_2),
\eqa
for some $p'(y_1|x_1y_2)$, or, equivalently, the channel is stochastically degraded.
\end{definition}

In the absence of receiver cooperation, a stochastically degraded interference channel is equivalent in its capacity to a physically degraded interference channel. As such, we will assume in the following that the channel is physically degraded, i.e., the DMZIC admits the Markov chain $X_2-(X_1,Y_2)-Y_1$. As a consequence, the following inequality holds
\begin{equation}
I(U;Y_2)\geq I(U;Y_1|X_1),\label{ieq:weak}
\end{equation}
for all input distributions $p(x_1)p(u)p(x_2|u)$.

The channel transition probability $p(y_1y_2|x_1x_2)$ becomes
\bqn
p(y_1y_2|x_1x_2)&=&p(y_2|x_1x_2)p(y_1|x_1x_2y_2)\\
&=&p(y_2|x_2)p(y_1|x_1y_2).
\eqn

The above definition of weak interference leads to the following sum capacity result.
\begin{theorem}\label{thm:DMZIC_CR}
The sum capacity of a DMZIC with weak interference as defined above is
\bqa\label{eq:DMZIC_SC}
\Cmat_{sum}=\max_{p(x_1)p(x_2)}\{I(X_1;Y_1)+I(X_2;Y_2)\}.
\eqa
\end{theorem}
\begin{proof}
This sum rate is achieved by two receivers decoding their own messages while treating any interference, if present, as noise.

For the converse,
\small
\bqn
&&\!\!\!\!\!\!n(R_1+R_2)-n \epsilon \\
&\stackrel{(a)}{\leq}&\!\!\!\!\!\!I(X_1^n;Y_1^n)+I(X_2^n;Y_2^n)\\
&\stackrel{(b)}{=}&\!\!\!\!\!\!\sum_{i=1}^n \left(H(Y_{1i}|Y_1^{i-1})-H(Y_{1i}|Y_1^{i-1}X_1^n)+H(Y_{2i}|Y_2^{i-1})\right.\\
&&\left.-H(Y_{2i}|Y_2^{i-1}X_2^n)\right)
\eqn
\bqn
&\stackrel{(c)}{\leq}&\!\!\!\!\!\!\sum_{i=1}^n \left(H(Y_{1i})-H(Y_{1i}|Y_1^{i-1}X_1^nY_2^{i-1})+H(Y_{2i}|Y_2^{i-1})\right.\\
&&\left.-H(Y_{2i}|Y_2^{i-1}X_{2i})\right)\\
&\stackrel{(d)}{=}&\!\!\!\!\!\!\sum_{i=1}^n\left(H(Y_{1i})-H(Y_{1i}|X_1^nY_2^{i-1})+I(X_{2i};Y_{2i}|U_i)\right)\\
&\stackrel{(e)}{=}&\!\!\!\!\!\!\sum_{i=1}^n\left(H(Y_{1i})-H(Y_{1i}|X_{1i}Y_2^{i-1})+I(X_{2i};Y_{2i}|U_i)\right)\\
&=&\!\!\!\!\!\!\sum_{i=1}^n (I(U_iX_{1i};Y_{1i})+I(X_{2i};Y_{2i}|U_i))\\
&=&\!\!\!\!\!\!\sum_{i=1}^n (I(X_{1i};Y_{1i})+I(U_i;Y_{1i}|X_{1i})+I(X_{2i};Y_{2i}|U_i))\\
&\stackrel{(f)}{\leq}&\!\!\!\!\!\!\sum_{i=1}^n (I(X_{1i};Y_{1i})+I(U_i;Y_{2i})+I(X_{2i};Y_{2i}|U_i))\\
&=&\!\!\!\!\!\!\sum_{i=1}^n (I(X_{1i};Y_{1i})+I(U_iX_{2i};Y_{2i}))\\
&\stackrel{(g)}{=}&\!\!\!\!\!\!\sum_{i=1}^n (I(X_{1i};Y_{1i})+I(X_{2i};Y_{2i})),
\eqn
\normalsize
where $U_i\triangleq Y_2^{i-1}$ for all $i$, $(a)$ follows the Fano's Inequality, $(b)$ is from the chain rule and the definition of mutual information, $(c)$ is because of the fact that conditioning reduces entropy, and that $Y_{2i}$ is independent of any other random variables given $X_{2i}$, $(d)$ is due to the memoryless property of the channel and the fact that $Y_{1i}$ is independent of any other random variables given $X_{1i}$ and $Y_{2i}$, then $(X_{1,i}^{n},Y_{1i})-(X_1^{i-1},Y_2^{i-1})- Y_1^{i-1}$ forms a Markov chain. By the weak union property, the Markov chain $Y_{1i}-(X_1^{n},Y_2^{i-1})-Y_1^{i-1}$ holds; $(e)$ is because of the Markov chain $(X_{1}^{i-1},X_{1,i+1}^n)- (X_{1i},Y_2^{i-1})- Y_{1i}$. The easiest way to prove it is using the \textit{Independence Graph}. Alternatively, we first note that the Markov chain
\bqn
(X_{1}^{i-1},X_{1,i+1}^n,Y_2^{i-1})- (X_{1i},Y_{2i})- Y_{1i}
\eqn
holds, since given $X_{1i}$ and $Y_{2i}$, $Y_{1i}$ is independent of $X_{1}^{i-1},X_{1,i+1}^n,Y_2^{i-1}$.
By the weak union property, the following Markov chain is obtained:
\bqn
(X_{1}^{i-1},X_{1,i+1}^n)- (X_{1i},Y_2^i)- Y_{1i}.
\eqn
Together with the Markov chain
\bqn
(X_{1}^{i-1},X_{1,i+1}^n)- X_{1i}-Y_2^i
\eqn
because of the independence between $Y_2^n$ and $X_1^n$, we get the Markov chain:
\bqn
(X_{1}^{i-1},X_{1,i+1}^n)- X_{1i}-(Y_{1i},Y_2^i)
\eqn
by the contraction property. Again, using the weak union property and then the decomposition property, we obtain the Markov chain
\bqn
(X_{1}^{i-1},X_{1,i+1}^n)- (X_{1i},Y_2^{i-1})- Y_{1i}
\eqn
as desired. Since $U_i$ and $X_{1i}$ are independent, then $p(x_1x_2u)=p(x_1)p(u,x_2)$, thus $(f)$ comes from (\ref{ieq:weak}). Finally, $(g)$ follows from the Markov chain $U_i- X_{2i}- Y_{2i}$. At last, by introducing a time-sharing random variable $Q$, one obtains
\bqn
R_1+R_2&\leq& I(X_1;Y_1|Q)+I(X_2;Y_2|Q)+\epsilon\\
&\leq& \max_{p(x_1)p(x_2)}\{I(X_1;Y_1)+I(X_2;Y_2)\}+\epsilon.
\eqn
\end{proof}

\textit{Remark $1$:} The Markov chain (\ref{eq:weakcond2}) is a sufficient, but not necessary, condition for the mutual information condition
\bqa
I(X_2;Y_1|X_1)\leq I(X_2;Y_2),\label{eq:DMZIC_weakextention}
\eqa
for all product input distribution on $\Xmat_1\times\Xmat_2$. One can find examples such that the mutual information condition holds but the Markov chain is not valid. This is different from that of the Gaussian case; it can be shown that the coefficient $a\leq 1$ in a Gaussian ZIC is a sufficient and necessary condition for (\ref{eq:DMZIC_weakextention}) to hold. It is yet unknown if condition (\ref{eq:DMZIC_weakextention}) is sufficient for the sum capacity result (\ref{eq:DMZIC_SC}) to hold for DMZIC with weak interference.

\subsection{Capacity Outer-bound for DMZIC with Weak Interference}\label{sec:equivalence}
For Gaussian ZICs with weak interference, Sato \cite{Sato:78IT} obtained an outer-bound using the capacity region of a related Gaussian broadcast channel constructed due to the equivalence in capacity between a GZIC with weak interference and a degraded GIC. The same technique can be used to obtain a capacity outer-bound for DMZIC with weak interference, i.e., that satisfies the Markov chain $X_2-(X_1,Y_2)-Y_1$. Specifically, for any such DMZIC with weak interference, one can find an equivalent (in capacity region) DMDIC whose capacity region is bounded by that of an associated degraded broadcast channel.

\begin{theorem}\label{thm:outer-bound}
For a DMZIC that satisfies the Markov chain $X_2-X_1Y_2-Y_1$, the capacity region is outer-bounded by
\small
\bqn
\Rmat_{OB}=\overline{co}\left\{\bigcup_{p(u)p(x_1x_2|u)}(R_1,R_2)\left|\begin{array}{l}R_1\leq I(U;Y_1),\\ R_2\leq I(X_1X_2;Y_2'|U)\end{array}\right.\right\},
\eqn
\normalsize
where $U-X_1X_2-Y_2'-Y_1$ forms a Markov chain and $\|\Umat\|=\min\{\|\Ymat_1\|,\|\Ymat_2'\|,\|\Xmat_1\|\cdot\|\Xmat_2\|\}$.
\end{theorem}
\begin{proof}
Suppose that the DMZIC with weak interference has inputs and outputs $X_1$, $X_2$ and $Y_1$, $Y_2$ respectively. Let us denote by  $X_1'$, $X_2'$ and $Y_1'$, $Y_2'$ the inputs and outputs of another DMIC. Set $X_1'=X_1$, $X_2'=X_2$, and $Y_1'=Y_1$ but define $Y_2'$ to be a bijection of $X_1$ and $Y_2$, denoted as $Y_2'=f(X_1,Y_2)$. As the Markov chain $(X_1',X_2')-Y_2'-Y_1'$ holds, the DMIC specified by the input pair $(X_1',X_2')$, and the output pair $(Y_1',Y_2')$ is indeed a DMDIC.
	
The proof that this DMDIC has the same capacity region as the specified DMZIC, and hence is outer-bounded by the associated broadcast channel  follows in exactly the same fashion as Costa's proof for the Gaussian case \cite{Costa:85IT}, hence is omitted here.
\end{proof}

\textit{Remark $2$:} The output $Y_2'$ need not necessarily be a bijection function of $X_1$ and $Y_2$; instead, depending on the transition probability $p(y_1|x_1y_2)$, other $Y_2'$ can be constructed. However, the associated broadcast channels would have the same the capacity region. It will become clear in the following example.

\subsection{Numerical Example}
\begin{example}
Let $\|\Xmat_1\|=\|\Xmat_2\|=\|\Ymat_1\|=\|\Ymat_2\|=2$ and the channel transition probability be given by
\bqn
p(y_1y_2|x_1x_2)=p(y_2|x_2)p(y_1|x_1y_2),
\eqn
where $p(y_2|x_2)$ and $p(y_1|x_1y_2)$ are specified in Table \ref{table:ex1}.
\begin{table}[htp]
\centering
\caption{Channel Transition Probabilities}\label{table:ex1}
		\begin{tabular}{|c|c|c||c|c|c|}
  \hline
   $p(y_2|x_2)$ & $y_2=0$ & $y_2=1$ & $p(y_1|x_1y_2)$ & $y_1=0$ & $y_1=1$\\\hline
	 $x_2=0$ & $.1$ & $.9$ & $x_1y_2=00$ or $11$ & $.75$ & $.25$\\\hline
	 $x_2=1$ & $.9$ & $.1$ & $x_1y_2=01$ or $10$ & $0$ & $1$\\\hline
		\end{tabular}
\end{table}

By Theorem \ref{thm:DMZIC_CR}, the sum capacity is
\bqn
\Cmat_{sum} = \max_{p(x_1)p(x_2)} I(X_1;Y_1)+I(X_2;Y_2)\approx.531.
\eqn
Moreover, one can construct $Y_2'$ as follows:
\bqn
Y_2'=\left\{\begin{array}{l}0,\textrm{     if }x_1y_2 = 00\textrm{ or }11,\\1, \textrm{   otherwise}.\end{array}\right.
\eqn
Then $p(y_2'|x_1x_2)$ is given in Table~\ref{table:y_2'}.
\begin{table}[htp]
\centering
\caption{$P(Y_2'|X_1X_2)$}\label{table:y_2'}
\begin{tabular}{|c|c|c|}
\hline
$p(y_2'|x_1x_2)$ & $y_2'=0$ & $y_2'=1$\\\hline
$x_1x_2=00$ & $.1$ & $.9$\\\hline
$x_1x_2=01$ & $.9$ & $.1$\\\hline
$x_1x_2=10$ & $.9$ & $.1$\\\hline
$x_1x_2=11$ & $.1$ & $.9$\\\hline
\end{tabular}
\end{table}

Using Theorem \ref{thm:outer-bound}, the capacity region of the DMZIC is outer-bounded by  that of the associated DMDBC:
\bqn  
\Rmat_{OB}=\overline{co}\left\{\bigcup_{p(u)p(x_1x_2|u)}(R_1,R_2)\left|\begin{array}{l}R_1\leq I(U;Y_1),\\ R_2\leq I(X_1X_2;Y_2'|U)\end{array}\right.\right\},
\eqn
If one takes the bijection function to construct $Y_2'$, it will lead to the same outer-bound.
If we fix $R_2$ to be $x$, then
\bqn
\max_{R_2=x}R_1 &=& \max_{H(Y_2'|U)=x+h_2(.1)} H(Y_1)-H(Y_1|U)\\
&\leq& \log(|\Ymat_1|)-f_T(x+h_2(.1)),
\eqn
where $f_T(\cdot)$ is a function defined by Witsenhausen and Wyner \cite{Witsenhausen&Wyner:75IT}.
Fig.~\ref{fig:outer-bound} depicts the new outer-bound specified by
\bqn
\Rmat_{OB}'=\left\{(R_1,R_2)|R_1\leq \log|\Ymat_1|-f_T(x+h_2(.1)), R_2\leq x\right\}.
\eqn
This new outer-bound significantly improves upon the following bound
\bqn
R_1&\leq&I(X_1;Y_1|X_2),\\
R_2&\leq&I(X_2;Y_2),\\
R_1+R_2&\leq&I(X_1;Y_1)+I(X_2;Y_2).
\eqn
\begin{figure}[htp]
\centerline{\leavevmode \epsfxsize=3.25in \epsfysize=2.65in
\epsfbox{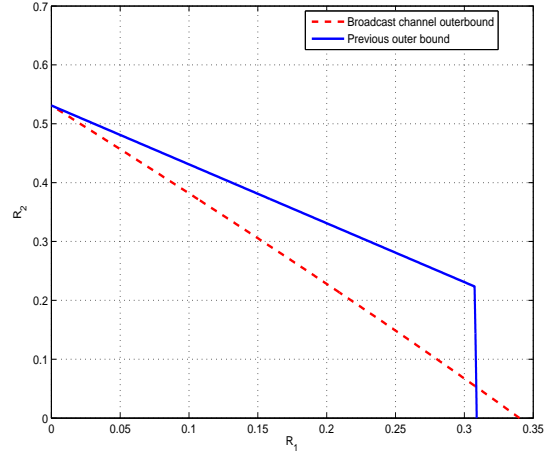}}
\caption{\label{fig:outer-bound} Comparison of the outer-bounds.}
\end{figure}
\end{example}

\section{The DMIC with Mixed Interference}\label{sec:mixed}
\begin{definition}
A DMIC is said to have \textit{mixed interference} if it satisfies the Markov chain
\small
\bqa\label{eq:mixed_cond1}
X_2-(X_1,Y_2)-Y_1
\eqa 
\normalsize and
\small
\bqa\label{eq:mixed_cond2}
I(X_1;Y_1|X_2)\leq I(X_1;Y_2|X_2)
\eqa
\normalsize
for all possible product distributions on $\Xmat_1\times\Xmat_2$.
\end{definition}
This definition is motivated by GIC with mixed interference, which can be shown to be equivalent in capacity region to a degraded GIC satisfying (\ref{eq:mixed_cond1}) by setting $p(y_1y_2|x_1x_2)=p(y_2|x_1x_2)p'(y_1|x_1y_2)$, where $p'(y_1|x_1y_2)$ is normal distribution with mean $(1-\sqrt{ab})x_1+\sqrt{a}y_2$ and variance $1-a$. 
The sum capacity for GIC with mixed interference was established in \cite{Motahari&Khandani:09IT}. We obtain a parallel result for the DMIC with mixed interference as defined above.
\begin{theorem}\label{thm:mixed_cr}
The sum capacity of the DMIC with mixed interference satisfying the two conditions (\ref{eq:mixed_cond1}) and (\ref{eq:mixed_cond2}) is
\small
\bqa
\max_{p(x_1)p(x_2)}\left\{I(X_2;Y_2|X_1)+\min\{I(X_1;Y_1),I(X_1;Y_2)\}\right\}.
\eqa
\end{theorem}
\normalsize
\begin{proof}
In order to achieve this sum rate, user $1$ transmits its message at a rate such that both receivers can decode it by treating the signal from user $2$ as noise; user $2$ transmits at the interference-free rate since receiver $2$ is able to subtract the interference from user $X_1$.

For the converse, we prove the following two sum rate bounds separately:
\small
\bqa
n(R_1+R_2)&\leq& \sum_{i=1}^nI(X_{1i}X_{2i};Y_{2i}),\label{eq:mixed_sc1}\\
n(R_1+R_2)&\leq& \sum_{i=1}^nI(X_{1i};Y_{1i})+I(X_{2i};Y_{2i}|X_{1i}).\label{eq:mixed_sc2}
\eqa
\normalsize
For (\ref{eq:mixed_sc1}), the derivation follows the same steps as Costa and El Gamal's result\cite{Costa&ElGamal:87IT}.
For (\ref{eq:mixed_sc2}), we use similar techniques for establishing the sum capacity of the DMZIC with weak interference in Section~\ref{sec:main}. First, notice that (\ref{eq:mixed_cond1}) implies
\small
\bqa
I(U;Y_1|X_1)\leq I(U;Y_2|X_1)\label{eq:mixed_cond3}
\eqa
\normalsize
for any $U$ whose joint distribution with $X_1,X_2,Y_1,Y_2$ is
\small
\bqa
p(u,x_1,x_2,y_1,y_2)=p(u)p(x_1x_2|u)p(y_1y_2|x_1x_2). \label{eq:u}
\eqa 
\normalsize Then,
\small
\bqn
&&\!\!\!\!\!\!n(R_1+R_2)-n \epsilon \\
&\stackrel{(a)}{\leq}&\!\!\!\!\!\!I(X_1^n;Y_1^n)+I(X_2^n;Y_2^n|X_1^n)\\
&=&\!\!\!\!\!\!\sum_{i=1}^n \left(H(Y_{1i}|Y_1^{i-1})-H(Y_{1i}|Y_1^{i-1}X_1^n)+H(Y_{2i}|Y_2^{i-1}X_1^n)\right.\\
&&\left.-H(Y_{2i}|Y_2^{i-1}X_2^nX_1^n)\right)\\
&\stackrel{(b)}{\leq}&\!\!\!\!\!\!\sum_{i=1}^n \left(H(Y_{1i})-H(Y_{1i}|Y_1^{i-1}X_1^nY_2^{i-1})+H(Y_{2i}|U_iX_{1i})\right.\\
&&\left.-H(Y_{2i}|X_{2i}X_{1i}U_i)\right)\\
&=&\!\!\!\!\!\!\sum_{i=1}^n (I(U_iX_{1i};Y_{1i})+I(X_{2i};Y_{2i}|U_iX_{1i}))\\
&=&\!\!\!\!\!\!\sum_{i=1}^n (I(X_{1i};Y_{1i})+I(U_i;Y_{1i}|X_{1i})+I(X_{2i};Y_{2i}|U_iX_{1i}))\\
&\stackrel{(c)}{\leq}&\!\!\!\!\!\!\sum_{i=1}^n (I(X_{1i};Y_{1i})+I(U_i;Y_{2i}|X_{1i})+I(X_{2i};Y_{2i}|U_iX_{1i}))\\
&\stackrel{(d)}{=}&\!\!\!\!\!\!\sum_{i=1}^n (I(X_{1i};Y_{1i})+I(X_{2i};Y_{2i}|X_{1i})),
\eqn
\normalsize
where $(a)$ is because of the independence between $X_1^n$ and $X_2^n$; $(b)$ is from the fact that conditioning reduces entropy and $U_i\triangleq (X_1^{i-1}X_{1,i+1}^n, Y_2^{i-1})$; $(c)$ is from (\ref{eq:mixed_cond3}); and $(d)$ is because the memoryless property of the channel and (\ref{eq:u}).
Finally, from (\ref{eq:mixed_sc1}) and (\ref{eq:mixed_sc2}), we have
\small
\begin{eqnarray}
R_1+R_2&\leq& \max_{p(x_1)p(x_2)} I(X_1X_2;Y_2),\nonumber\\
R_1+R_2 &\leq& \max_{p(x_1)p(x_2)}I(X_1;Y_1)+I(X_2;Y_2|X_1),\nonumber
\end{eqnarray}
respectively.
\normalsize
\end{proof}

We give the following example where the obtained sum capacity helps determine the capacity region of a DMIC.

\begin{example}
	Consider the following deterministic channel:
	\small
	\bqn
	Y_1 &=& X_1\cdot X_2,\\
	Y_2 &=& X_1 \oplus X_2,
	\eqn
	\normalsize
	where the input and output alphabets $\Xmat_1$, $\Xmat_2$, $\Ymat_1$ and $\Ymat_2 = \{0,1\}$. Notice that this channel does not satisfy the condition of the deterministic interference channel in \cite{ElGamal&Costa:82IT}.
	Obviously, the Markov chain (\ref{eq:mixed_cond1}) holds. Moreover,
	\small
	\bqn
	I(X_1;Y_1|X_2)\!\!\!\!\!\!&=&\!\!\!\!\!\!H(Y_1|X_2)= p(x_2=1)H(X_1),\\
	I(X_1;Y_2|X_2)\!\!\!\!\!\!&=&\!\!\!\!\!\!H(Y_2|X_2) = H(X_1).
	\eqn
	\normalsize
	Therefore,
	\small
	\bqn
	I(X_1;Y_1|X_2)\leq I(X_1;Y_2|X_2),
	\eqn
	\normalsize
	for all possible input product distributions on $\Xmat_1\times\Xmat_2$.
	Therefore, this is a DMIC with mixed interference. Apply Theorem \ref{thm:mixed_cr}, we compute the sum capacity is
	\small
	\bqn
	\Cmat_{sum}\!\!&=&\!\!\max_{p(x_1)p(x_2)}\left[\min(I(X_1X_2;Y_2),I(X_1;Y_1)+I(X_2;Y_2|X_1))\right]\\
	&=&1.
	\eqn
\normalsize
Given that $(1,0)$ and $(0,1)$ are both trivially achievable, the above sum capacity leads to the capacity region for this DMIC to be $\{(R_1,R_2):R_1+R_2\leq 1\}$.
\end{example}
\section{Conclusion}\label{sec:conclusion}
In this paper, we derived the sum capacity for the DMZICs with weak interference where weak interference is defined using a Markov condition. Similar techniques are then applied to derive the sum capacity for DMIC with mixed interference. Both results are analogous to the sum capacity results for the corresponding Gaussian channel, both in the expression of the capacity and in the encoding schemes that achieve the capacity.

The weak interference condition is defined using a Markov chain, as opposed to that using the mutual information inequality. While it appears to be somewhat restrictive, it is not known whether the definition using the mutual information condition will lead to the same sum capacity result.

\end{document}